\documentclass[twoside]{IEEEtran}

\usepackage{amsmath,amssymb}
\usepackage{bm}
\usepackage{ascmac}

\usepackage{color}

\usepackage{amsmath}
\usepackage{amsthm}
\usepackage{amssymb}
\usepackage{algorithm}
\usepackage{algorithmic}

\usepackage{graphicx}  % for arXiv.org test
\newtheorem{theorem}{Theorem}

\newtheorem{lemma}{Lemma}

\newtheorem{definition}{Definition}

\begin{document}

\title{Information-theoretically secure equality-testing protocol with dispute resolution}

\author{Go Kato, Mikio Fujiwara, and Toyohiro Tsurumaru%
\thanks{
Go Kato and
Mikio Fujiwara are with National Institute of Information and Communications Technology (NICT), Nukui-kita, Koganei, Tokyo 184-8795, Japan (e-mail: go.kato@nict.go.jp, fujiwara@nict.go.jp).
Toyohiro Tsurumaru is with Mitsubishi Electric Corporation, Information Technology R\&D Center,
5-1-1 Ofuna, Kamakura-shi, Kanagawa, 247-8501, Japan (e-mail: Tsurumaru.Toyohiro@da.MitsubishiElectric.co.jp).
}
}

\maketitle

\begin{abstract}
There are often situations where two remote users each have data, and wish to (i) verify the equality of their data, and (ii) whenever a discrepancy is found afterwards, determine which of the two modified his data.
The most common example is where they want to authenticate messages they exchange.
Another possible example is where they have a huge database and its mirror in remote places, and whenever a discrepancy is found between their data, they can determine which of the two users is to blame.

Of course, if one is allowed to use computational assumptions, this function can be realized readily, e.g., by using digital signatures.
However, if one needs information-theoretic security, there is no known method that realizes this function efficiently, i.e., with secret key, communication, and trusted third parties all being sufficiently small.

In order to realize this function efficiently with information-theoretic security, we here define the ``equality-testing protocol with dispute resolution'' as a new framework.
The most significant difference between our protocol and the previous methods with similar functions is that we allow the intervention of a trusted third party when checking the equality of the data.
In this new framework, we also present an explicit protocol that is information-theoretically secure and efficient.

\end{abstract}


\section{Introdution}

There are often situations where two remote users each have data which are supposed to be the same, and they want to (i) verify that their data are indeed the same, and (i) whenever a discrepancy between their data is found afterwards, determine which of the two actually modified his data.
Here, we consider ``equality-testing protocols with dispute resolution'' (or ET protocols for short) for realizing this function and study its information-theoretic security.

The most straightforward application of this protocol would be where the two party want to authenticate messages they exchange, and also prevent further change.
Another possible example is where the two users have a huge database and its mirror in remote places.
If this type of protocols is realized, they will be able to confirm the equality of the data on both sides, and moreover, whenever a discrepancy is found afterwards, they can determine which of the two users is to blame.

Of course, if one is allowed to use computational assumptions, our ET protocol can be realized readily, e.g., by using digital signatures.
However, if one needs information-theoretic security, there is no known method for realizing it efficiently, i.e., with sufficiently small amounts of secret keys and communication, and with at most one trusted third party (TTP).
At first glance, this function seems feasible by a straightforward use of the $A^2$ code \cite{333869,10.1007/3-540-48329-2_29,10.1007/3-540-39118-5_15,10.1007/3-540-49264-X_25} or the unconditionally secure digital signature schemes (USDS) \cite{10.1007/3-540-44448-3_11,SwansonStinson+2016+35+67}, but these methods consume enormous resources (communication and secret key) and thus not practical.
More precisely, both the communication length and the secret key length of these methods must exceed the data length, so they become virtually impossible when one needs to handle larger data, such as an entire data center (see section \ref{sec:background_of_definition}).

Here we rigorously define a framework for ET protocols which admits information-theoretic security and efficient implementation.
The most significant difference between our protocol and the previous methods with similar functions (i.e., $A^2$ codes and USDSs) is that we allow the intervention of a TTP in the equality-testing phase, which corresponds to message authentication tag generation in the previous methods (see \ref{sec:TTP_intervention} section for the detail).
Then in this generalized framework, we also present an explicit protocol that is information-theoretically secure and efficient.
Namely, we present a protocol achieving
\begin{align}
&{\rm lengths\ of\ communication\ and\ secret\ key}\nonumber\\
&=O\left(\log (1/\epsilon)\cdot\log r\right)
\label{eq:parameter_order_ETP}
\end{align}
with $r$ being the data length and $\epsilon$ being the security parameter (success probability of a malicious player).
Note that parameter region (\ref{eq:parameter_order_ETP}) indeed overcomes the limitations of the aforementioned previous methods.

We note that our ET protocols are particularly useful in and compatible with quantum key distribution networks (QKDNs) \cite{itut_y3800, Sasaki:11}, for the following two reasons.
First, a QKDN normally has a key management authority, which can be used as the TTP for our protocol.
Second, our protocols need to be supplied with a new secret key each time they are executed.
The only ways to fulfill this requirement, at least at present, are either to use the so-called trusted courier or to use a QKDN.

Note that this implies that our protocols greatly enhance the functionality of QKDNs.
It is often thought that QKDNs can only be used for limited purposes, namely, one-to-one secret communication and authentication.
Our protocols, however, indeed provide more versatile cryptographic functionalities, such as the message authentication with non-repudiation, or the mirroring of huge databases with modification prevention, mentioned at the beginning of this section.

Finally we note that our ET method can also be seen as an improvement of the method proposed and implemented in Ref. \cite{9650707} which realizes ITS message authentication function using QKDN, secret sharing, and TTPs.
In this method, tamper resistance and dispute resolution functions are implemented using two types of TTPs: a shared calculator and a verifier.
Our method here realizes a similar function with a reduced number of TTPs, i.e., with only one TTP, and benefits in terms of cost savings and scalability.

\section{Definition of equality-testing protocol with dispute resolution}

\subsection{Setting and the definition of the protocol}

Suppose that a trusted third party (TTP) and two players, Alice and Bob, are connected to each other by unauthenticated public channels\footnote{Channels where messages are neither encrypted nor authenticated.}.
Also suppose that Alice and Bob each holds data $m^A, m^B\in\{0,1\}^r$, which are supposed to be equal.
Our goal here is to (i) verify the equality of $m^A$ and $m^B$, and (ii) whenever a discrepancy is found afterwards, determine with certainty which of the two actually modified his data.
To this goal we define the following type of protocols.
\begin{definition}[Equality-testing protocol with dispute resolution (ET protocol, for short), Fig. \ref{fig:ET_schematic}]
\label{def:ET_protocol}
A equality-testing protocol with dispute resolution is where Alice and Bob check the equality of their data with the help of the TTP, and consists of the following three phases.
\begin{enumerate}
\item {\bf Key-distribution phase}: The TTP distributes secret keys to Alice and Bob \footnote{They share secret keys by using some method other than the public channels.}.
\item {\bf Equality-testing phase}: The TTP communicates with Alice and Bob, and verifies the equality of their data (whether $m^A=m^B$ or not).
If the equality is confirmed, the TTP announces ``success"; otherwise, it announces ``failure".
\item {\bf Dispute-resolution phase}:
After a successful completion of the equality-testing phase, if there is a dispute between Alice and Bob about the equality of their data, the TTP arbitrates as follows:

The TTP receives data $m^{A*}, m^{B*}$ from Alice and Bob respectively, which they claim to be correct.
Then he compares them with the communication content of the equality-testing phase, and announces one of the following: ``both are correct", ``$m^{A*}$ is correct", ``$m^{B*}$ is correct" or ``undecidable".
\end{enumerate}
\end{definition}

Throughout the paper, whenever the TTP ``announces'' something, it means that he sends the same message to both Alice and Bob simultaneously.

The ``dispute'' that triggers the dispute-resolution phase can occur, e.g., when an outsider (other than Alice, Bob, or the TTP) retrieves the same part of Alice's and Bob's data respectively, which should equal, but finds a discrepancy.

\begin{figure}[htbp]
\begin{center}
 \includegraphics[width=0.8\linewidth]{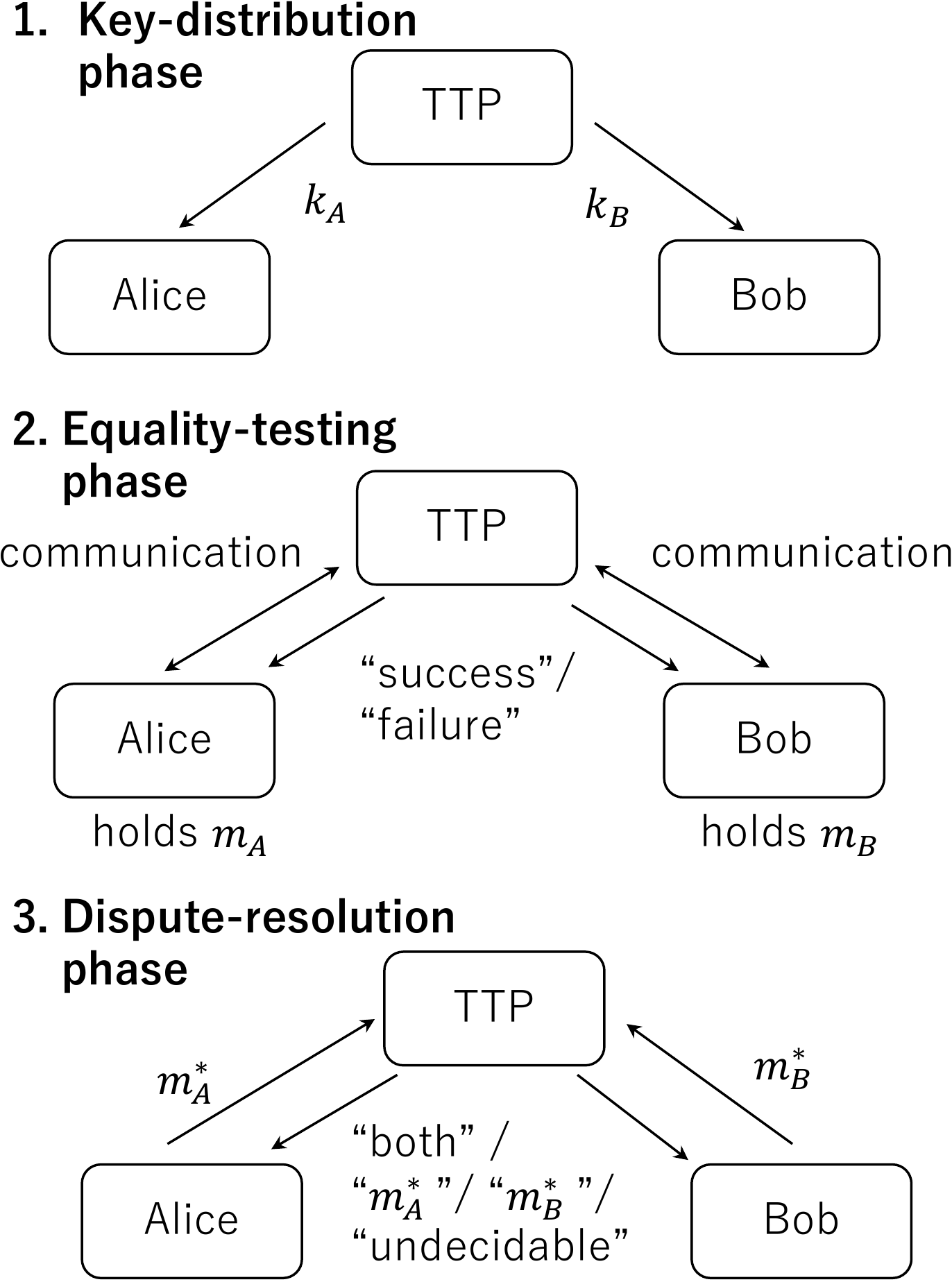}
 \caption{Conceptual diagram of our ET protocol.}
 \label{fig:ET_schematic}
\end{center}
\end{figure}

\subsection{Security criteria}

The basic concept behind our security is that ``honest players lose nothing.''
The situations can be classified into the following three cases.
\begin{itemize}
\item If both Alice and Bob are honest, their claims are always accepted (soundness).
\item If only one of Alice or Bob is honest, then the honest one will lose with a negligible probability ($\epsilon$-unmodifiable).
\item If both Alice and Bob are malicious (i.e., not honest), they gain nothing; i.e., our protocols are not responsible for this case.
\end{itemize}
Of these three cases, the first two needs to be addressed.
We define these cases rigorously as follows.

\begin{definition}[Security criteria of our ET protocols]
\label{def:security}
We say that an ET protocol is $\epsilon$-secure if it satisfies the following two conditions.
\begin{itemize}
\item {\bf Soundness}: If both Alice and Bob are honest and their data are the same ($m^A=m^B$), then the equality-testing phase always succeeds.
\item {\bf $\epsilon$-unmodifiability}:
Suppose that Alice is honest and the TTP announces ``success'' in the equality-testing phase.
Then in the subsequent dispute-resolution phase, the TTP announces an outcome unfavorable to Alice (namely, ``$m^{B*}$ is correct'' or ``undecidable'') with a probability $\le\epsilon$.

In addition, the same condition also holds with the roles of Alice and Bob being exchanged.
\end{itemize}
\end{definition}

This definition does not prohibit malicious players from intentionally disabling the equality-testing phase.
However, this does not mean that the security is compromised.
Indeed, there is no protocol from which one can expect more, as can be seen as follows:
Suppose, for example, that Alice is honest and Bob is malicious, and Bob performs the equality-testing phase correctly as specified, though by using a wrong data $m^{B*}(\ne m^A)$.
From the perspective of the TTP, this situation cannot be distinguished from the one in which the ``honest'' Bob has the ``correct'' data $m^{B*}$, and the ``malicious" Alice tries to claim a ``wrong'' data $m^A$.
In such situation, and if the equality-testing phase cannot fail and thus Bob can trigger the dispute-resolution, Alice will be judged malicious  with a significant probability of $1/2$, even though she is actually honest.

\subsection{Background to the definition above}
\label{sec:background_of_definition}

Next, we explain the background that led to the above settings and security criteria.

\subsubsection{We want to be able to handle huge data}
\label{sec:handle_huge_data}
Even when the data is huge (e.g., genome data ($r\ge 10^9$) or the entire data center ($r\ge 10^{15}$)), the resource comsumption (namely, communication volume, and secret key length) should be sufficiently smaller than $r$, so that the protocol can be executed efficienly in practice.

\subsubsection{Dispute-resolution phase as a deterrent}

It is natural that Alice and Bob disclose their data $m_A^*, m_B^*$ themselves in the dispute-resolution phase.
However, if this happens frequently, the average communication length per protocol execution will be enormous, contradicting the requirements of the preceding paragraph.

Therefore, we here assume that the actual frequency of dispute-resolution is sufficiently low, and thus the communication length for the dispute-resolution can be ignored when evaluating the performance of the protocol.
(On the other hand, the consumption of secret keys will always be taken into account).
Note that this evaluation criterion is the same as that of conventional non-repudiation MACs (e.g., \cite{10.1007/3-540-44448-3_11}), so it is by no means a weakness specific to our method.

In order to justify this evaluation criterion, we will focus on situations where ``though the dispute-resolution is not frequent, once it actually happens and the fraud is discovered, the damage will be enormous.''
In other words, we assume that dispute-resolution is a deterrent and that it will not be executed frequently.
For example, if the data is an official document or a will, any modification to it would immediately mean an illegal act, and if discovered, would inevitably result in legal punishment\footnote{Conversely, our protocols are not suitable for situations where a malicious player can easily escape before an dispute resolution occurs.}.

\subsubsection{Straightforward use of MAC does not work}
\label{sec:MAC_does_not_work}

Next, we will make comparisons with the existing methods.

For the sake of simplicity, we first ignore resources.
If Alice and Bob only needs to check the equality of their data, it suffices for them to exchange message authentication code (MAC) tags (Fig.\ref{fig:ET_by_DS}(a))\footnote{Bob consider $m^B$ in our scheme to be the message received from Alice, and verify its MAC. Alice also does the same.}.
In addition if they also wish to determine which of the two, Alice or Bob, actually modified the data (i.e., if they want $\epsilon$-unmodifiability), they can do so by using a MAC scheme equipped with a non-repudiation function (hereafter referred to as non-repudiation MACs).

However, non-repudiation MACs are either limited in usage or consume too much resource to be practical.
Indeed, as we require information theoretic security (ITS) here, we cannot use those schemes based on computational assumptions, such as the widely used digital signature schemes (see, e.g., Ref. \cite{katz2020introduction}, Chapter 13).
There are also non-repudiation MACs achieving ITS, such as $A^2$ code \cite{333869, 10.1007/3-540-48329-2_29,10.1007/3-540-39118-5_15, 10.1007/3-540-49264-X_25} and unconditionally secure digital signature (USDS) \cite{10.1007/3-540-44448-3_11,SwansonStinson+2016+35+67}, but they are subject to the following restrictions,
\begin{align}
({\rm communication\ of\ the\ equality}&\ {\rm testing\ phase}\ge)\nonumber\\
\ {\rm MAC\ tag\ length}&>{\rm data\ length}\ r,
\label{eq:communication_length_gt_data_length}\\
{\rm MAC\ key\ length}&>{\rm data\ length}\ r,
\label{eq:key_length_gt_data_length}
\end{align}
again contradicting the requirements of Section \ref{sec:handle_huge_data}.

Inequality (\ref{eq:communication_length_gt_data_length}) can readily be proved for non-repudiation MACs in general.
On the other hand, Inequality (\ref{eq:key_length_gt_data_length}) is an empirical relation which seems to hold for all the practical non-repudiation MAC schemes that we are aware of (see Refs. \cite{333869, 10.1007/3-540-48329-2_29, 10.1007/3-540-39118-5_15,10.1007/3-540-49264-X_25, 10.1007/3-540-44448-3_11, SwansonStinson+2016+35+67} and references therin).

\begin{figure}[htbp]
\begin{center}
 \includegraphics[width=0.8\linewidth]{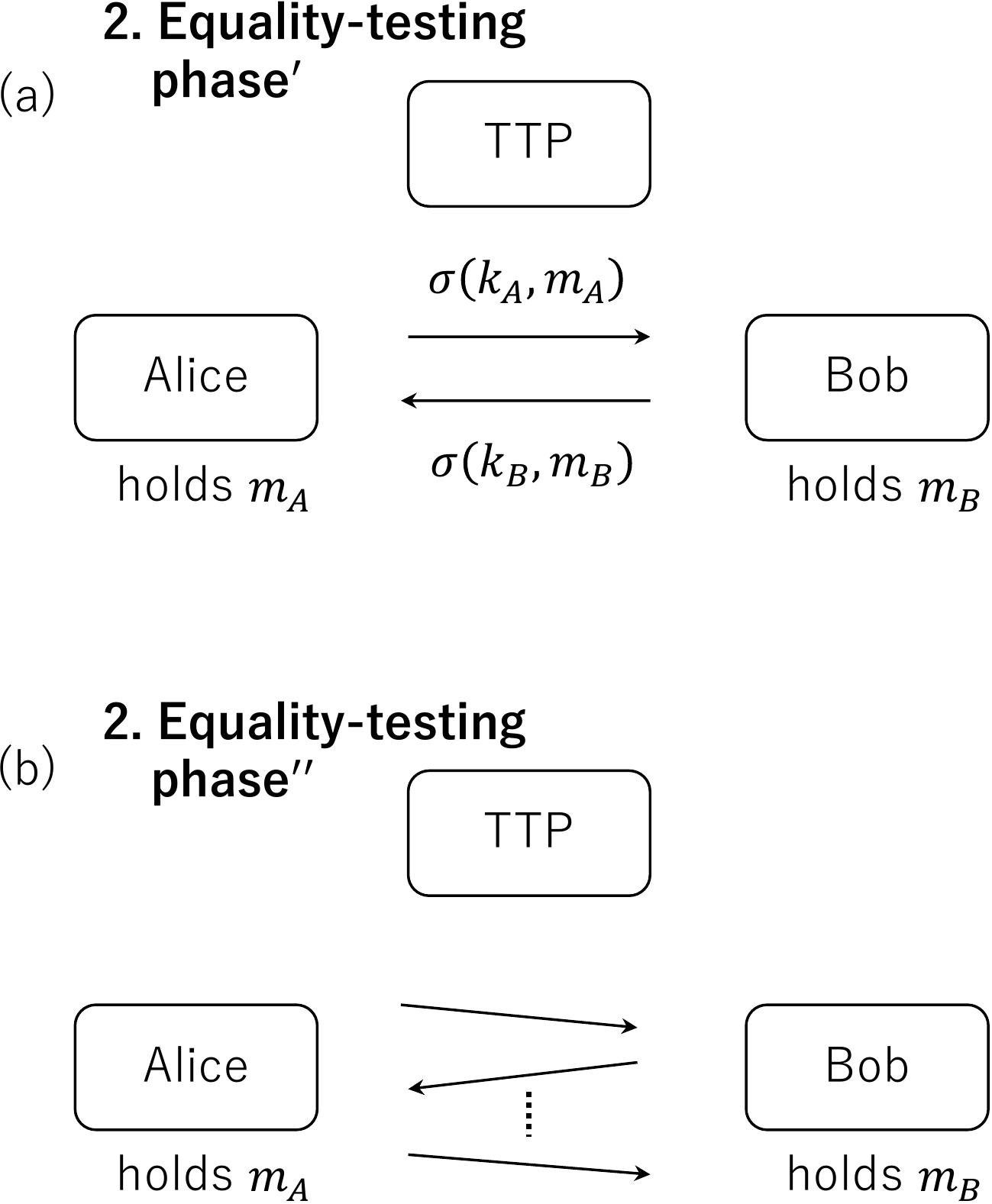}
 \caption{(a) The situation where Alice and Bob exchange MAC tags $\sigma(k^A,m^A), \sigma(k^B,m^B)$ to realize the equality-testing phase, without the help of TTP.
(b) More general situation where the TTP does not intervene in the equality-testing phase.
}
 \label{fig:ET_by_DS}
\end{center}
\end{figure}

\subsubsection{TTP's intervention is necessary in the equality-testing phase}
\label{sec:TTP_intervention}
Moreover, it can easily be shown that Inequality (\ref{eq:communication_length_gt_data_length}) in fact holds for any ET protocol in which the TTP does not intervene in the equality-testing phase (Fig.\ref{fig:ET_by_DS}(b)).
Therefore, in order to avoid restriction (\ref{eq:communication_length_gt_data_length}), we need to generalize the setting by letting the TTP intervene in the equality-testing.

We consider this to be a natural generalization:
All the existing non-repudiation MAC schemes allow the TTP's intervention in the key-distribution and the dispute-resolution phases, so there is no reason to prohibit it only in the equality-testing phase.

\subsection{Compatibility with quantum key distribution networks}
Furthermore, our ET protocols are particularly compatible with the quantum key distribution network (QKDN) \cite{itut_y3800, Sasaki:11}.
The reasons are as follows.

First, the generalized setting where the TTP intervenes in the equality-testing phase, described in the previous section, can easily be realized in QKDNs.
This is because a QKDN normally has a key management authority, which is a TTP that is always in operation, and thus suitable for the present purpose.

Second, as is always the case with information-theoretic cryptographic protocols in general, this protocol needs to be supplied with a new secret key each time .
In order to achieve this, at least at present, one must either use the so-called trusted courier (i.e., the TTP himself physically delivers a storage medium containing the key) or use quantum key distribution.

\subsection{Nontrivial part of the problem}
\label{sec:nontriviality}
Once one accepts that the TTP can intervene in the equality-testing phase, the simplest construction might seem to be the one in which both Alice and Bob send $m^A, m^B$ and/or the corresponding MAC tags to the TTP.
However, such constructions is again found impractical by a similar reasoning as in Section \ref{sec:MAC_does_not_work}:
If a MAC scheme without the non-repudiation function is used there, the protocol can indeed detect if there was a modification at all, but cannot determine which of the two, Alice or Bob, modified his/her data (no $\epsilon$-nonmodifiability).
Hence a non-repudiation MAC scheme with ITS is necessary, but that again makes the protocol subject to Inequalities (\ref{eq:communication_length_gt_data_length}) and (\ref{eq:key_length_gt_data_length}), contradicting the requirements of Section \ref{sec:handle_huge_data}.

Therefore, we need a construction different from the conventional non-repudiation MAC schemes.

\section{Efficient ET protocol with information-theoretic security}

With the above setting and security criteria, we propose a secure and efficient protocol that is free from restrictions (\ref{eq:communication_length_gt_data_length}) and (\ref{eq:key_length_gt_data_length}), and achieves (\ref{eq:parameter_order_ETP});
hence it can be used in the parameter region
\begin{equation}
{\rm communication\ length},\ {\rm key\ length}<<{\rm data\ length}\ r.
\end{equation}

Our protocol is specified in Table ``Protocol \ref{protocol_main}.''
\begin{Protocol}[htbp]
\caption{Efficient ET protocol with information-theoretic security}
\label{protocol_main}
\begin{algorithmic}

\STATE {\bf Key-distribution phase}:
The TTP distributes secret keys as follows.

\STEPONE Randomly select $n$ distinct numbers out of $\{1,\dots,N\}$, and denote them by $\Omega$ (i.e. $\Omega\subset_{\rm R}[N]$, with $[N]:=\{1,2,\cdots,N\}$).
\STEPTWO Randomly select ``equality-testing keys'' $k^A_{\rm et}=(k^A_1,k^A_2,\cdots,k^A_{N})$, $k^B_{\rm et}=(k^B_1,k^B_2,\cdots,k^B_{N})$ with $k^A_i,k^B_j\in\{0,1\}^l$,
such that $k^A_j=k^B_j$ for $j\in \Omega$.

(For example: First choose all $k^A_i$'s according to the uniform distribution, then let $k^B_j=k^A_j$ for $j\in \Omega$, and then choose the undecided elements of $k^B_j$'s according to the uniform distribution.)
\STEPTHREE Randomly select ``secure communication keys'' $k^A_{\rm sc}, k^B_{\rm sc}\in\{0,1\}^{l_{\rm sc}}$.
\STEPFOUR Send $k^A_{\rm et}, k^A_{\rm sc}$ to Alice, and $k^B_{\rm et}, k^B_{\rm sc}$ to Bob.
\end{algorithmic}
\begin{algorithmic}
\STATE
\STATE $\dagger$ All subsequent communications must be authenticated in an information theoretic manner; e.g., by using MAC scheme specified by Lemma \ref{lem:ITSMAC} and consuming part of secure communication keys $k^A_{\rm sc}, k^B_{\rm sc}$.

\STATE
\STATE {\bf Equality-testing phase}
\STEPONE
Alice calculates hash values $s^A_i:=f(k^A_i,m^A)$ ($i=1,\dots,N$) by using a hash function $f$ specified by Lemma \ref{lem:Hash_func}.
Then she encrypts them (by using the one-time pad consuming part of $k^A_{\rm sc}$) and sends it to the TTP.

Bob also does the same.
\STEPTWO
The TTP announces ``success'' if $s^A_j=s^B_j$ for $\forall j\in\Omega$ holds; otherwise she announces ``failure''.
\end{algorithmic}
\begin{algorithmic}
\STATE
\STATE {\bf Dispute-resolution phase}

\STEPONE 
Alice (Bob) sends data $m^{A*}$ ($m^{B*}$) to the TTP.

\STEPTWO
The TTP performs the following checks:

With $g^{\alpha\beta}:=\big|\{j\in U|f(k^\alpha_j,m^\beta)=s^\alpha_j\}\big|$ for $\alpha,\beta\in\{A,B\}$,
\begin{enumerate}
\item Announce ``Both is correct'' if $m^{A*}=m^{B*}$ and finish.
\item Announce ``$m^{A*}$ is correct'' if $g^{AA}=N\;\land\;(g^{BA} > g^{AB}\;\lor\;g^{BB}< N)$, and finish.

Also perform the same check with indices $A,B$ exchanged.
\item Announce ``undecidable.''
\end{enumerate}
\end{algorithmic}
\end{Protocol}
The function $f$ appearing there should satisfy the condition of the following lemma.
\begin{lemma}[Almost universal$_2$ hash function using polynomials, Ref. \cite{10.1007/3-540-55719-9_77}, or Ref. \cite{katz2020introduction}, Theorem 4.17]
\label{lem:Hash_func}
There exists a function $f:(k,m)\mapsto s$ with
$k\in\{0,1\}^l, m\in\{0,1\}^r, s\in\{0,1\}^l$ satisfying the following:
With variable $K\in\{0,1\}^l$ being uniformly distributed, and for arbitrary distinct pair of data $m\neq m'\in\{0,1\}^r$, we have
\begin{align}
\Pr\left[f(K,m)= f(K,m')\right]&\leq \left\lceil\frac r{  l}-1\right\rceil2^{-l}.
\end{align}
\end{lemma}
Also, all the communications in the equality-testing and dispute-resolution phases should be authenticated in an information theoretic manner; e.g., by using the following MAC scheme and consuming part of pre-distributed keys.
\begin{lemma}[Information-theoretically secure message authentication code (ITS-MAC). See, e.g, Ref. \cite{katz2020introduction}, Theorems 4.17 and 4.25]
\label{lem:ITSMAC}
There exists an ITS-MAC scheme satisfying the following:
The MAC tag $t\in\{0,1\}^{n+l}$ is generated using a function $g:(k,m)\mapsto t$ from a message $m\in\{0,1\}^{r'}$ with the uniformly distributed key $k\in\{0,1\}^{2(n+l)}$, and it achieves $\left\lceil\frac{r'}{n+l}-1\right\rceil2^{-(n+l)}$-security (i.e., the adversary can forge the MAC tag $t$ only with a probability $\le\left\lceil\frac{r'}{n+l}-1\right\rceil2^{-(n+l)}$).
\end{lemma}

Then this protocol satisfies the following security.
\begin{theorem}
\label{the:main}
For the data length $r\ge 256$, and for the security parameter $\epsilon\le 2^{-4}$,
Protocol \ref{protocol_main} is $\epsilon$-secure, if we choose its parameters as
\begin{eqnarray}
n&=&\lceil3\log_2(16/\epsilon)\rceil,\\
l&=&\lceil\log_2 r\rceil,\\
N&=&2n,\\
l_{\rm sc}&=&4nl+16(n+l).
\end{eqnarray}
In this case, the total length of secret keys ($k^A_{\rm et}, k^A_{\rm sc}, k^B_{\rm et}, k^B_{\rm sc}$) is $8(nl+2n+2l)$ bits.
The total communication length of the equality-testing phase achieves $4(nl+2n+2l+1)$ bits, and that of the dispute-resolution phase is $2(r+4n+4l+2)$ bits.
\end{theorem}

Hence it indeed satisfies condition (\ref{eq:parameter_order_ETP}), anticipated in Introduction.
For example, even when the data pair is 1-Pbit long ($r=2^{50}\ge10^{15}$), we can achieve $10^{-12}$-security, with the communication length in the equality-testing phase $\le64$kbit and the total secret key length $\le32$kbit.

Note that, if we let $N=2, n=1$ in Protocol~\ref{protocol_main} and also let $\Omega$ fixed at $\{2\}$, our protocol becomes an example of the straightforward construction mentioned at the beginning of Section \ref{sec:nontriviality}.
However, this example uses a MAC scheme without the non-repudiation function, and thus cannot achieve $\epsilon$-nonmodifiability, as already discussed in the same section.
In the present method, instead of using a non-repudiation MAC, we let $N$ and $n$ be sufficiently large, and choose $\Omega$ randomly.
This improves the TTP's ability to detect malicious actions by Alice or Bob during the equality-testing phase.

\section{Proof of Theorem \ref{the:main}}
Of the two conditions presented in Definition \ref{def:security}, the soundness is evident, so we prove $\epsilon$-unmodifiability only.
We begin by proving the following idealized case.
\begin{lemma}
\label{lmm:Thm_with_ideal_channels}
Suppose that Alice, Bob and the TTP can use ideal secure channels whenever necessary in the equality-testing and the dispute-resolution phases\footnote{In other words, they can use channels which are completely free of tampering in the equality-testing and the dispute-resolution phases, and in addition, those channels used in Step 1 of the equality-testing phase are free of eavesdropping.}
(thus they do not need to consume the secure communication keys $k^A_{\rm sc}, k^B_{\rm sc}$).

Also suppose that $r\ge2^8$ and $\epsilon_1\le2^{-8}$.
Then Protocol \ref{protocol_main} is $\epsilon_1$-unmodifiable with its parameters chosen as $n=\lceil3\log_2(1/\epsilon_1)\rceil$, $N=2n$, $l=\lceil\log_2 r\rceil$. 
\end{lemma}

\subsection{Proof of Lemma \ref{lmm:Thm_with_ideal_channels}}
We will prove this lemma in two steps.

\begin{lemma}
Under the same setting as in Lemma \ref{lmm:Thm_with_ideal_channels}, Protocol \ref{protocol_main} is $\epsilon_2$-unmodifiable, where
\begin{align}
\epsilon_2
&:=\max_{t\in\{n,\dots, N\}}\frac{(N-n)!t!}{N!(t-n)!}
\sum_{u=t-n}^{N-n}{N-n \choose u}q^{u}(1-q)^{N-n-u},\\
q&:=\left\lceil\frac r{  l}-1\right\rceil2^{-l}.
\end{align}
\end{lemma}

\begin{proof}
Since the protocol is invariant under the permutation of the roles of Alice and Bob, it suffices to show for the case where Alice is honest and Bob is malicious.

In this case, Alice always submits the correct hash value $s^A_i=f(k^A_i,m^A)$ and the correct data $m^A$ to the TTP.
On the other hand, according to Definition \ref{def:security}, Bob's goal is to (i) submit some hash value $s^{B*}$ and succeed in the equality-testing phase, knowing $m^A$ and $k^B$,
and then (ii) submit some $m^{B*}(\ne m^A)$ in the dispute-resolution phase and let the TTP announce ``$m^{B*}$ is correct'' or ``undecidable.''

Note that in the equality-testing phase, Bob is informed only of the result, ``success'' or ``failure,'' and that he needs to submit $m^{B*}$ only when it was ``success.''
Thus his success probability does not change even if he fixes $m^{B*}$ in advance in the equality-testing phase.
Therefore, we may modify Bob's goal as follows.
\begin{itemize}
\item[] {\bf Malicious Bob's goal}:
(Knowing $m^A$, $k^B$) choose the values of both $s^{B*}$ and $m^{B*}(\ne m^A)$ in advance, then submit them to the TTP, and let the TTP announce ``success'' in the equality-testing phase, and ``$m^{B*}$ is correct'' or ``undecidable'' in the dispute-resolution phase.
\end{itemize}

In light of the description of Protocol \ref{protocol_main}, the above goal is equivalent to selecting $s^{B*}$ and $m^{B*}(\ne m^A)$ satisfying
\begin{eqnarray}
s_j^{B*}&=&f(k^B_j,m^A)\ {\rm for}\  j\in\Omega,
\label{eq:Bob_attack_condition0A}
\\
g^{BB}&=&N,
\label{eq:Bob_attack_condition0B}
\\
g^{AB}&\ge&g^{BA},
\label{eq:Bob_attack_condition0C}
\end{eqnarray}
where we used the fact that $k^A_j=k^B_j$ for $j\in\Omega$ in deriving (\ref{eq:Bob_attack_condition0A}).
By noting that condition (\ref{eq:Bob_attack_condition0B}) means $s^{B*}_i=f(k^B_i,m^{B*})$ for all $i=1,\dots,N$, we can further rewrite these conditions as
\begin{align}
&f(k^B_j,m^{B*})=f(k^B_j,m^A)\ {\rm for}\ \forall j\in\Omega,
\label{eq:Bob_attack_condition1}\\
&\left|\left\{\,j\,|f(k^A_j,m^A)=f(k^A_j,m^{B*})\right\}\right|\ge t,
\label{eq:Bob_attack_condition2}
\end{align}
where
\begin{equation}
t:=\left|\left\{\,j\,|f(k^B_j,m^A)=f(k^B_j,m^{B*})\right\}\right|.
\label{eq:t_defined}
\end{equation}

Below we will evaluate the probabilities of conditions (\ref{eq:Bob_attack_condition1}) and (\ref{eq:Bob_attack_condition2}).

Condition (\ref{eq:Bob_attack_condition1}) says that $t$ subscripts of (\ref{eq:t_defined}) (which are uniquely determined by $m^{B*}$) include $\Omega$.
Bob must select those $t$ subscripts (by selecting $m^{B*}$) without knowing $\Omega$.
Hence (\ref{eq:Bob_attack_condition1}) holds with a probability
\begin{equation}
p_{\rm et}(t)={N\choose t}{t\choose n}\left({N\choose t}{N\choose n}\right)^{-1}=\frac{(N-n)!t!}{N!(t-n)!}.
\label{eq:p_et}
\end{equation}

Condition (\ref{eq:Bob_attack_condition2}) demands that $f(k^A_j,m^A)=f(k^A_j,m^{B*})$ holds for more than $t-n$ subscripts $j\in\{1,\dots,N\}\setminus\Omega$.
However, Bob does not know keys $k^A_j$ corresponding to those indices $j\in\{1,\dots,N\}\setminus\Omega$, since they are generated independently of $k_j^B$'s (cf. comment inside parentheses in Protocol \ref{protocol_main}, step 1).
Due to this fact and Lemma \ref{lem:Hash_func}, condition (\ref{eq:Bob_attack_condition2}) holds only with a probability
\begin{equation}
p_{\rm dr}(t)\le\sum_{u=t-n}^{N-n}{N-n \choose u}q^u(1-q)^{N-n-u}.
\label{eq:p_arb}
\end{equation}

For a fixed value of $t$, Bob's success probability is upper bounded by $p_{\rm et}(t)p_{\rm dr}(t)$.
Thus we have the lemma.
\end{proof}

\begin{lemma}
For $r\ge2^8$, $\epsilon_1\le2^{-8}$, 
$n=\lceil3\log_2(1/\epsilon_1)\rceil$, $N=2n$, and $l=\lceil\log_2 r\rceil$, we have $\epsilon_2\le\epsilon_1$.
\end{lemma}

\begin{proof}
Since $N=2n$, $p_{\rm et}(t)$ of (\ref{eq:p_et}) can be bounded as
\begin{eqnarray}
p_{\rm et}(t)&=&\frac{n!t!}{(2n)!(t-n)!}\nonumber\\
&=&\frac{t}{2n}\cdot\frac{t-1}{2n-1}\cdots\frac{t-(n-1)}{2n-(n-1)}\nonumber\\
&\le&\left(\frac{t}{2n}\right)^n,
\end{eqnarray}
where we used that fact that $\frac{t-s}{2n-s}\le\frac{t}{2n}$ for $0\le s\le t\le 2n$.

By Theorem 11.1.4 of Ref. \cite{cover2012elements}, $p_{\rm dr}$ of (\ref{eq:p_arb}) can be bounded as
\begin{eqnarray}
p_{\rm dr}(t)&\le& (2n-t+1)2^{-nD(p||q)},\\
D(p||q)&=&p\log_2\frac{p}{q}+(1-p)\log_2\frac{1-p}{1-q},\\
p&=&t/n-1.
\end{eqnarray}

If we choose $n$ as specified by the lemma, we have for $t\le3n/2$,
\begin{equation}
p_{\rm et}(t)\le\epsilon_1.
\label{eq:bound_p_et}
\end{equation}
On the other hand for $t> 3n/2$, we have $p\ge1/2$ and also $q\le 1/8$ due to $r\ge256$.
Then we have $D(p||q)\ge 1/2$ and thus
\begin{equation}
p_{\rm dr}(t)\le(n/2+1)2^{-n/2}\le (n/2+1)2^{-n/6}\cdot \epsilon_1\le \epsilon_1,
\label{eq:bound_p_dr}
\end{equation}
where the last inequality follows by noting that $n\ge3\log_2(1/\epsilon_1)\ge24$.

Combining (\ref{eq:bound_p_et}) and (\ref{eq:bound_p_dr}), we obtain the lemma.
\end{proof}

\subsection{Proof of Theorem \ref{the:main}}

Unlike the ideal situation of Lemma \ref{lmm:Thm_with_ideal_channels}, in the actual situation described in Definition \ref{def:ET_protocol} one can only use unauthenticated public channels.
There Alice, Bob, and the TTP must use information-theoretic MAC and/or one-time pad encryption wherever necessary.

First, communications in Step 1 of the equality-testing phase must be secret in an information-theoretic sense.
This can be realized by using the one-time pad (OTP) encryption, and consumes $4nl$ bits of the secret key.

In addition, all eight communication rounds in the equality-testing and the dispute-resolution phases must be authenticated in an information-theoretic manner.
Since the length $r'$ of the content of each round satisfies $r'\le\max\{r,2nl\}$, one can authenticate each round with $\epsilon_1$-security by using the MAC scheme of Lemma \ref{lem:ITSMAC}.
Thus for all eight rounds in total, one can achieve the overall authenticity with $8\epsilon_1$-security by consuming secret keys of $16(n+l)$-bits.

By combining the secrecy and authenticity thus realized with the result of Lemma \ref{lmm:Thm_with_ideal_channels}, and by letting $\epsilon_1=\epsilon/16$, we can achieve the $\epsilon$-nonmodifiability in the actual situation described in Definition \ref{def:ET_protocol}.

The breakdown of secret key consumption here is:
$4nl$ bits for the equality testing keys $k^A_{\rm et}, k^B_{\rm et}$, and $4nl+16(n+l)$ bits for the secure communication keys $k^A_{\rm sc}, k^B_{\rm sc}$.
The latter consists of $4nl$ bits for the OTP encryption of $s^A_i,s_B^j$, and $16(n+l)$ bit for generating MAC tags for the eight rounds of communication.
All of these add up to $8(nl+2n+2l)$ bits.

The communication length of the equality-testing phase consists of $4nl$ bits for $s^A_i,s_B^j$, four bits for the TTP's announcements, and $8(n+l)$ bits for the MAC tags for these four rounds of communication.
These add up to $4(nl+2n+2l+1)$ bits.

Similarly, for the dispute-resolution phase there are $2r$ bits of communication for $m^{A*}$ and $m^{B*}$, two bits for the TTP's announcements, and $8(n+l)$ bits for the MAC tags for these four rounds, all of which add up to $2(r+4n+4l+2)$ bits.

\section{Summary and outlook}

We proposed a new type of cryptographic protocols called ``equality-testing protocol with dispute resolution'' (ET protocols) and also presented an explicit protocol that is information-theoretically secure and efficient.
Our ET protocols enable two remote users each having data to (i) verify the equality of their data, and (ii) whenever a discrepancy is found afterwards, determine which of the two modified his data.
The ET protocols are particularly useful in and compatible with quantum key distribution networks (QKDNs), and can also greatly enhance the functionality of QKDNs.

A possible future work is to reduce the amount of communication needed in the dispute-resolution phase.
It will also be interesting to actually implement this type of protocols in one of real QKDNs.

\section*{Acknowledgment}
M.F. and T.T. were supported in part by ``ICT Priority Technology Research and Development Project'' (JPMI00316) of the Ministry of Internal Affairs and Communications, Japan.
This work was supported in part by JSPS KAKENHI Grant Numbers JP18H05237, JP20K03779, JP21K03388.

\bibliographystyle{IEEEtran}
\bibliography{DigitalEndorsement}

\begin{IEEEbiographynophoto}{Go Kato} was born in Japan, in 1976.
He received the M.S. and Ph.D. degrees in science from The University of Tokyo in 2001 and 2004, respectively.
 From 2004 to 2022, he worked with the NTT Communication Science Laboratories, NTT Corporation, as a Scientist. In 2022, he joined NICT (National Institute of Information and Communications Technology), as a research manager.
He has been engaged in the theoretical investigation of quantum information. He is especially interested in mathematical structures emerging in the field of quantum information. He is a member of the Physical Society of Japan.

\end{IEEEbiographynophoto}

\begin{IEEEbiographynophoto}{Mikio Fujiwara} received the B.S. and M.S. degrees in electrical engineering and the Ph.D. degree in physics from Nagoya University, Nagoya, Japan, in 1990, 1992, and 2002, respectively. He has been involved R\&D activities at NICT (previous name CRL,  Ministry of Posts and Telecommunications of Japan) since 1992.
\end{IEEEbiographynophoto}

\begin{IEEEbiographynophoto}{Toyohiro Tsurumaru} was born in Japan in 1973.
He received the B.S. degree from the Faculty of Science, University of Tokyo, Japan in 1996,
and M.S. and Ph.D. degrees in physics from the Graduate School of Science, University of Tokyo, Japan in 1998 and 2001, respectively.
Then he joined Mitsubishi Electric Corporation in 2001.
His research interests include theoretical aspects of quantum cryptography and of modern cryptography.
\end{IEEEbiographynophoto}

\end{document}